\documentclass[review]{elsarticle}

\usepackage{amsmath,amsthm,amssymb,amscd,amstext,amsfonts}
\usepackage{mathtools}
\usepackage{mathrsfs}
\usepackage{thmtools}
\usepackage{latexsym}
\usepackage{verbatim}
\usepackage{framed}
\usepackage{graphicx}
\usepackage{stmaryrd}
\usepackage{enumerate}
\usepackage{fullpage}
\usepackage{bm}
\usepackage{color}
\usepackage{hyperref}
\usepackage{url}
\usepackage{physics}
\usepackage{multicol}
\usepackage{tikz-cd}
\usetikzlibrary{decorations.pathmorphing}   
\usetikzlibrary{arrows}
\usepackage{makeidx}
\makeindex
\usepackage{thm-restate}    
\usepackage{cleveref}   

\usepackage{caption} 
\declaretheorem[name=Theorem,numberwithin=section]{thm}
\declaretheorem[name=Lemma,numberwithin=section]{lemma}
\declaretheorem[name=Corollary,numberwithin=section]{corollary}
\declaretheorem[name=Definition,numberwithin=section]{definition}
\declaretheorem[name=Remark,numberwithin=section]{remark}

\DeclareMathOperator{\NE}{NE}
\DeclareMathOperator{\RNE}{RNE}
\DeclareMathOperator{\AND}{AND}
\DeclareMathOperator{\OR}{OR}

\DeclareMathOperator{\INTS}{INTS}             
\DeclareMathOperator{\IP}{IP}                       
\DeclareMathOperator{\EQ}{EQ}                       

\newcommand{\N}{\mathbb{N}}

\newcommand{\cA}{\mathcal A}

\newcommand{\cM}{\mathcal M}

\newcommand{\cO}{\mathcal O}
\newcommand{\cP}{\mathcal P}

\usepackage{lineno,hyperref}
\modulolinenumbers[5]
\allowdisplaybreaks

\journal{Journal of Computer and System Sciences}

\newcommand{\vertc}{|\hspace{-1.5mm}c}        
\DeclareMathOperator{\lang}{L}     
\DeclareMathOperator{\ttm}{T}     
\DeclareMathOperator{\ssp}{S}    
\DeclareMathOperator{\ts}{TS}     

\begin{document}

\begin{frontmatter}

\title{Lifting query complexity to time-space complexity for two-way finite automata\tnoteref{mytitlenote}}
\tnotetext[mytitlenote]{ Some results in this paper were pointed out in our conference paper (in the Proceedings of 6th International Conference on Theory and Practice of Natural Computing,  LNCS 10687, pp. 305-317, 2017).}



\author[mymainaddress1]{Shenggen Zheng\corref{eqc}}
\address[mymainaddress1]{Peng Cheng Laboratory, 
Shenzhen 518055, China}
\cortext[eqc]{These authors contributed equally: Shenggen Zheng, Yaqiao Li\\Corresponding author: Shenggen Zheng}
\ead{zhengshg@pcl.ac.cn}


\author[mymainaddress]{Yaqiao Li\corref{eqc}}
\address[mymainaddress]{Shenzhen MSU-BIT University, Shenzhen, 518172, China}
\ead{yqlipublic@gmail.com}
\author[mymainadd1]{Minghua Pan}
\address[mymainadd1]{Guangxi Key Laboratory of Cryptography and Information Security, Guilin University of Electronic Technology, Guilin 541004, China}

\author[mymainaddress3]{Jozef Gruska}
\address[mymainaddress3]{ Faculty of Informatics, Masaryk University, Brno 60200, Czech Republic}
\ead{gruska@fi.muni.cz}

\author[mymainaddress4]{Lvzhou Li}
\ead{lilvzh@mail.sysu.edu.cn}
\address[mymainaddress4]{ School of Computer Science and Engineering, Sun Yat-Sen University, Guangzhou 510006, China}


\begin{abstract}
Time-space tradeoff has been studied in a variety of models, such as Turing machines, branching programs,  and finite automata, etc. While communication complexity as a technique has been applied to study finite automata, it seems it has not been used to study time-space tradeoffs of finite automata. We design a new technique showing that separations of query complexity can be lifted, via communication complexity, to separations of time-space complexity of two-way finite automata. As an application, one of our main results exhibits the first example of a language $L$ such that  the time-space complexity of two-way probabilistic finite automata with a bounded error (2PFA) is $\widetilde{\Omega}(n^2)$, while of exact two-way quantum finite automata with classical states (2QCFA) is $\widetilde{O}(n^{5/3})$, that is, we demonstrate for the first time that exact quantum computing has an advantage in time-space complexity comparing to classical computing.
\end{abstract}



\end{frontmatter}



\renewcommand*{\thefootnote}{\arabic{footnote}}



\section{Introduction}

An important way to get deeper insights into the power of various quantum resources and operations is to explore the power of various quantum variations of the basic models of classical computation. Of a special interest is to do that for various quantum variations of the classical finite automata, especially for those that use limited amounts of quantum  resources. It has been proved \cite{Kla00}  that exact 1-way quantum finite automata with classical states (1QCFA) \cite{2QCFA,ZhgQiu112} have no state (space) complexity  advantage over deterministic finite autumata in recognizing a language \cite{Klk03}. One of the main results in this paper shows that exact two-way quantum finite automata with classical states (2QCFA) \cite{2QCFA}, however, do have complexity advantage  over two-way probabilistic finite  automata (2PFA). We will show this advantage via the time-space complexity (i.e, the  product  of  the  time  and  the  logarithm  of  the number of states).

Time-space tradeoff is an important topic that dates back at least as early as  \cite{TStradeoff66} where time-space tradeoff for Turing machines were studied. It is often too difficult to obtain separation results for time complexity or space complexity alone, hence in practice one alternative is to study the product of time and space, so-called the time-space tradeoff. This has been studied in a variety of models, such as Turing machines and branching programs \cite{BNS92, BST98, cctoBP3},  finite automata \cite{TStradeoff2DFA81},  and  other specific computational problems \cite{Klk03, KSW07}, etc.

Early time-space tradeoff results are mostly obtained via various combinatorial methods, such as \cite{ tradeoffdiag,TStradeoff66, TStradeoff2DFA81} etc. One strong tool for studying the time-space tradeoff for branching programs is via communication complexity, see e.g., \cite{BNS92,BST98,cctoBP3}. The time-space tradeoff for language recognition and automata have been studied before, such as \cite{ automataTradeoff3,TStradeoff2DFA81,automataTradeoff2}, however, communication complexity was not involved. For time complexity alone, consider the non-regular language $L=\{x^ny^n \ | \ n\ge 1\}$. It is well-known that 2DFA cannot recognize $L$, 2PFA recognizes $L$ with exponential time \cite{2PFAlb,2PFAub}, while 2QCFA needs only polynomial time \cite{2QCFA}. This separation is also shown  via combinatorial analysis.  On the other hand, communication complexity has been used to study languages \cite{autoCC86}. Communication complexity has also been applied to study finite automata, though  previous studies are concerned on space complexity of the (classical and quantum)  automata, such as \cite{GRUSKA2017Generalizations,autoCC09,autoCC97,autoCC01,autoCC03,autoCC14}. 

In this work, we show how to use communication complexity to obtain time-space complexity results for two-way finite automata. We develop a technique,  namely  ``lifting technique'',  that automatically lifts separations of query complexity to separations of time-space complexity of two-way finite automata. With this we immediately obtain separation results for finite automata from known separations in query model. 
Our technique can also be viewed as an application of    lifting theorems in communication complexity to the study of time-space complexity of finite automata. 

Below we briefly discuss the lifting technique, and summarize its applications and limitations.

\subsection{The lifting technique}   \label{sec:technique}

Let 
    $h: \{0,1\}^p \to \{0,1\}$
be a Boolean function. Let $D(h)$ and $Q(h)$ denote the deterministic query complexity and the quantum query complexity of $h$, respectively (see Definition in Section \ref{sec:Query}). Suppose we know a separation in the query complexity world,
    \begin{equation}   \label{eq:sep-query}
        D(h) \gg Q(h),
    \end{equation}
where $\gg$ means that 
    $Q(h) = o(D(h))$.
Let $\ts_{2DFA}(L)$ and $\ts_{2QCFA}(L)$ denote the time-space complexity of recognizing language $L$ by a 2DFA and a 2QCFA, respectively (see Definition in Section \ref{def:2FA}). We will define a language $L$ depending on $h$ and show that 
    \begin{equation}   \label{eq:lift}
        \ts_{2DFA}(L) \gtrsim D(h) \quad \text{and} \quad
        \ts_{2QCFA}(L) \lesssim Q(h),
    \end{equation}
where the notation $\gtrsim$ and $\lesssim$ hide certain parameters. \eqref{eq:sep-query} and \eqref{eq:lift} together imply that
    \begin{equation}   \label{eq:sep-TS}
        \ts_{2DFA}(L) \gg \ts_{2QCFA}(L).
    \end{equation}
In this way, we ``lift'' a separation result for query complexity \eqref{eq:sep-query} to a separation for time-space complexity for two-way finite automata \eqref{eq:sep-TS}. Below we give some detail on \eqref{eq:lift}.

Let the alphabet $\Sigma = \{0,1, \#\}$. Every two-party Boolean function $f: \{0,1\}^n \times \{0,1\}^n \to \{0,1\}$ induces a language $L_f(n)$ defined as follows
    \begin{equation}  \label{eq:L-f}
        L_{f}(n) = \{x\#^ny  \in \Sigma^*\ | \ x, y \in \{0,1\}^n, f(x,y) = 1\}.
    \end{equation}

On one hand, we show that two-way finite automata recognizing $L_f(n)$ can be translated to communication protocols solving $f$.   Let $D^{cc}(f)$ and $R^{cc}(f)$ denote the deterministic communication complexity of $f$ and randomized communication complexity of $f$ with a bounded error, respectively. 

\begin{restatable}{thm}{thmlb} \label{thm:lb}
    $\ts_{2DFA}(L_f(n)) \ge n(D^{cc}(f) - 1)$,  
    $\ts_{2PFA}(L_f(n)) \ge n(R^{cc}(f) - 1)$.
\end{restatable}

Given a function 
    $h: \{0,1\}^p \to \{0,1\}$ 
and a two-party function 
    $g: \{0,1\}^m \times \{0,1\}^m \to \{0,1\}$, 
they naturally define another two-party function 
    $f: \{0,1\}^{pm} \times  \{0,1\}^{pm} \to \{0,1\}$
as follows,
    \begin{equation}  \label{eq:def-composition}
        f(x,y) = h(g(x_1,y_1), g(x_2,y_2), \ldots, g(x_p,y_p)),
    \end{equation}
where 
    $(x_i, y_i) \in  \{0,1\}^m \times \{0,1\}^m$. 
We denote \eqref{eq:def-composition} by 
    $f = h \circ g$, 
i.e., $f$ is a composed function of $h$ and $g$. 

Let $\IP_m$ denote the two-party inner product function 
    $\IP_m(a,b)=\sum_{j=1}^m a_j b_j \mod 2$, 
where 
    $a,b\in \{0,1\}^m$.  
Applying appropriate lifting theorems in communication complexity, Theorem \ref{thm:lb} implies the following corollary which corresponds to the first part of \eqref{eq:lift}. Let $R(h)$ denote the randomized query complexity of $h$. Let $\ts_{2PFA}(L)$ denote the time-space complexity of recognizing language $L$ by a  2PFA.

\begin{restatable}{corollary}{corlb} \label{cor:lb}
    Let
        $h: \{0,1\}^p \to \{0,1\}$
    and
        $f = h \circ \IP_m$
    where $m= \Theta(\log p)$. 
    Let
        $n=pm$
    be the input size for $f$,
    then
    \begin{equation}  \label{eq:det-lifting}
        \ts_{2DFA}(L_f(n)) \ge \widetilde{\Omega}(n D(h))
    \end{equation}
    and
    \begin{equation}  \label{eq:random-lifting}
        \ts_{2PFA}(L_f(n)) \ge \widetilde{\Omega}(n R(h)).
    \end{equation}
\end{restatable}

\begin{remark}
The reason for choosing $g= \IP_m$ is that to prove Corollary \ref{cor:lb} we need to use some lifting theorems in communication complexity which requires choosing $g=\IP_m$ (or some other functions satisfying certain properties, but $\IP_m$ is a simple and valid choice). We refer the interested reader to  \cite{chattopadhyay2019query,chattopadhyay2019simulation}.
\end{remark}

On the other hand, we show that quantum query algorithms for $h$ can be translated to two-way finite automata recognizing $L_f(n)$ where 
    $f = h \circ g$. 
Let $Q_{\epsilon}(h)$ denote the quantum query complexity of $h$ where the error is at most $\epsilon$. Let $Q_E(h)$ denote the quantum query complexity of $h$ where there is no error.  Let $\ts_{2QCFA,exact}(L)$ denote the time-space complexity of recognizing $L$ by a 2QCFA with no error. The following theorem corresponds to the second part of \eqref{eq:lift}.

\begin{restatable}{thm}{thmub} \label{thm:ub}
    Let 
        $f = h \circ g$ 
    be of the form \eqref{eq:def-composition}, let 
        $n=pm$ 
    be the input size for $f$. Then,
    \begin{equation}   \label{eq:Q-ub}
        \ts_{2QCFA}(L_f(n)) = O\Big(Q_\epsilon(h) \cdot n \cdot (\log p + m)\Big),
    \end{equation}
    and the error for the 2QCFA is at most $\epsilon$.  In particular, setting $\epsilon = 0$ and $m=\Theta(\log p)$, then 
    \begin{equation}   \label{eq:Qexact-ub}
        \ts_{2QCFA,exact}(L_f(n)) = O\Big(Q_E(h) \cdot n \cdot \log n\Big) = \widetilde{O} \Big(n Q_E(h) \Big).
    \end{equation}
\end{restatable}

\begin{remark}
    This is a good place to discuss the definition \eqref{eq:L-f}. For Boolean function 
        $f: \{0,1\}^n \times \{0,1\}^n \to \{0,1\}$,
    one could define 
        $L_{f}(q) = \{x\#^qy  \in \Sigma^*\ | \ x, y \in \{0,1\}^n, f(x,y) = 1\}$,
    and then optimize over $q$ to get the largest separation. Let us consider $2DFA$ and $2QCFA$ as an example. With this new definition, ignoring lower order $\log$ terms, the result for Corollary \ref{cor:lb} becomes 
        $\ts_{2DFA}(L_f(q)) \ge \widetilde{\Omega}(q \cdot D(h))$,
    and the result for Theorem \ref{thm:ub} (see details in Theorem \ref{thm:m-is-1}) becomes 
        $\ts_{2QCFA}(L_f(q)) \le \widetilde{O}((q+n) \cdot Q_\epsilon(h))$.
    From Corollary \ref{cor:lb}, since 
        $n=pm$
    and $m=\Theta(\log p)$,
    we have
        $p = \widetilde{\Theta}(n)$.
    So, suppose we have a function $h$ for which
        \begin{equation}    \label{eq:sep_query}
            D(h) = \widetilde{\Omega}(n^{\alpha}), \quad
            Q_\epsilon(h) = \widetilde{O}(n^{\beta}), \quad
            \text{for  some }
            0 \le \beta \le \alpha \le 1.
        \end{equation}
    Then,
        $\ts_{2DFA}(L_f(q)) \ge \widetilde{\Omega}(q \cdot n^{\alpha})$
    and
        $\ts_{2QCFA}(L_f(q)) \le \widetilde{O}((q+n) \cdot n^{\beta})$.
    With a simple calculation, one finds that the largest separation is achieved at 
        $q=n$,
    in which case
        \begin{equation}    \label{eq:sep_TS}
            \ts_{2DFA}(L_f(q)) \ge \widetilde{\Omega}(n^{1+\alpha}),    \quad
            \ts_{2QCFA}(L_f(q)) \le \widetilde{O}(n^{1+\beta}).
        \end{equation}

\end{remark}

\subsection{Separation results for time-space complexity}   \label{sec:sep-results}

Let 
    $h: \{0,1\}^p \to \{0,1\}$
be a Boolean function.
Let
    $f = h\circ \IP_m$
where 
    $m=\Theta(\log p)$.
Let 
    $n=pm$
be the input size for $f$. Consider the language 
    $L_f(n)$
defined in \eqref{eq:L-f}. As described in Section \ref{sec:technique}, we apply Corollary \ref{cor:lb} and Theorem \ref{thm:ub} together to lift a separation for the query complexity of $h$ to a separation for the time-space complexity of $L_f(n)$. Using known separation results for total Boolean functions, we obtain the following separations for time-space complexity for two-way finite automata.

\begin{restatable}{thm}{thmSepViaQuery} \label{thm:separation-via-query}
    In each of the following cases, there exist a total function $f: \{0,1\}^n \times \{0,1\}^n \to \{0,1\}$ (the function $f$ is different for each case),  such that 
    \begin{enumerate}[(1)]
    \item $\ts_{2DFA}(L_f(n)) = \widetilde{\Theta} (n^2)$ and  $\ts_{2QCFA}(L_f(n)) \le \widetilde{O}(n^{5/4})$.
    
    \item $\ts_{2DFA}(L_f(n)) = \widetilde{\Theta} (n^2)$ and  $\ts_{2QCFA,exact}(L_f(n)) \le \widetilde{O}(n^{3/2})$.
    
    \item $\ts_{2PFA}(L_f(n)) = \widetilde{\Theta} (n^2)$ and  $\ts_{2QCFA}(L_f(n)) \le \widetilde{O}(n^{4/3})$.
    
    \item $\ts_{2PFA}(L_f(n)) = \widetilde{\Theta}(n^2)$ and  $\ts_{2QCFA,exact}(L_f(n)) \le \widetilde{O}(n^{5/3})$.
    \end{enumerate}
\end{restatable}

However, the functions $h$ involved in the  Theorem  \ref{thm:separation-via-query} are all quite complicated. In Section \ref{sec:sep-results-proof} we give two other simple examples with weaker separations, in these examples the languages are explicit.

\begin{remark}
    Instead of total functions, one could consider partial functions $h$. 
    Since the lifting theorem we use (e.g., \cite{chattopadhyay2019query}) also holds when $h$ is a \emph{partial} function, and the simulation theorem Theorem \ref{thm:ub} also holds (it will become clear in the proof of Theorem \ref{thm:m-is-1}) when $h$ is a partial function,  we can also obtain separations of time-space complexity when $f$ are partial functions. For example, \cite{Forrelation} constructs a partial function $h$ on $n$ bits such that
        $R(h) \ge \widetilde{\Omega}(\sqrt{n})$
    and
        $Q_\epsilon(h) \le O(1)$.
    By adapting the discussion in \eqref{eq:sep_query} and \eqref{eq:sep_TS}, we can show that there is a partial function $f$ such that
        $\ts_{2PFA}(L_f(n)) = \widetilde{\Theta} (n^{3/2})$ 
    and  
        $\ts_{2QCFA}(L_f(n)) \le \widetilde{O}(n)$.
    However, by making a change of variable $n=N^{4/3}$, this separation is just as strong as (3) of Theorem \ref{thm:separation-via-query}. For separations of $Q_E$ from $D$ or $R$, we are not aware of separations by partial functions that are stronger than total functions, as a result, we can not improve (2) or (4) of Theorem \ref{thm:separation-via-query} by partial functions.
\end{remark}

\subsection{The limitations of the lifting technique}    \label{sec:lift-limitation}

A minor limitation of the lifting technique is that we require $m=\Theta(\log p)$ in Corollary \ref{cor:lb}, this restriction renders the lower bound to be of the form $\widetilde{\Omega}$ instead of $\Omega$, i.e., we always lose a $\log n$ factor in the lower bound obtained in this way. This limitation comes from the     current limitation of lifting theorems in communication complexity, which requires that $m \ge \Omega(\log p)$. 

Another limitation of the lifting technique is the following. Different versions of query complexity, such as $D(h), R(h), Q_E(h)$ and $Q(h)$, are known to be polynomially related for total functions (see, e.g., \cite{BdW02}). In fact, Aaronson{\it et al.~}\cite{aaronson2020quantum} showed that $D(h) \le O(Q(h)^4)$ for every $h$. Hence, among 2DFA, 2PFA, 2QCFA and 2QCFA without error, the separation of time-space complexity  obtainable from the above lifting method could only be $\Omega(n^2)$ and $O(n^\alpha)$ for some $\alpha \ge 5/4 > 1$, if we use only total functions. In particular, since $D(h) \le O(R(h)^3)$ (see \cite{BdW02}), the above lifting method could prove, at most, that some language $L$ satisfies $\ts_{2DFA} \ge \Omega(n^2)$ and $\ts_{2PFA}(L) \le O(n^{4/3})$. However, this limitation might possibly be overcome by using better separations of partial functions.

The next theorem shows that sometimes a better separation for time-space complexity can be obtained  by \emph{directly} simulating communication protocols using two-way finite automata. Let 
\begin{equation}   \label{eq:L-EQ}
L_{\EQ}(n) = \{x\#^ny \ | \ x, y \in \{0,1\}^n, x=y\}.
\end{equation}
This language separates the time-space complexity of 2DFA and 2PFA with tight bounds.

\begin{restatable}{thm}{thmLEQ} \label{thm:sep-L-EQ}
    $\ts_{2DFA}(L_{\EQ}(n)) = \Theta(n^2)$,  $\ts_{2PFA}(L_{\EQ}(n)) = \Theta(n\log n)$.
\end{restatable}

Observe that the $\EQ_n(x,y)$ can be viewed as a composed function of the form \eqref{eq:def-composition} where $h$ is the $\AND_n$ function and $g$ is the negation of the xor function, in particular $m=1$. Because $m=1 \neq \Theta(\log n)$, we can not apply Corollary \ref{cor:lb} to obtain the lower bound for $\ts_{2DFA}(L_{\EQ}(n))$ from the lower bound on the query complexity of $\AND_n$. 
Instead, we need to apply Theorem \ref{thm:lb} directly with a lower bound on the communication complexity of $\EQ_n$.

The rest of the paper is organized as follows. Section \ref{sec:def} defines  formally  all the related two-way finite automata,  time-space complexity, quantum query algorithms and the two-party communication model.  Section \ref{sec:Simulation} proves the theorems discussed in Section \ref{sec:technique} and Theorem \ref{thm:sep-L-EQ}. Section \ref{sec:sep-results-proof} proves Theorems \ref{thm:separation-via-query}. The last section discusses our method and some open problems. 

\section{Definitions}  \label{sec:def}

\subsection{Two-way finite automata and time-space complexity} \label{def:2FA}

We assume familiarity with (1-way) deterministic finite automaton (DFA), which is denoted by a tuple $(S, \Sigma, \delta, s_0, S_{acc}, S_{rej})$, where $S$ is a finite set of \emph{states}, $\Sigma$ is a finite set of input symbols called the \emph{alphabet}, $\delta$ is a \emph{transition function}, $s_0 \in S$ is an \emph{initial state}, and $S_{acc}, S_{rej}\subseteq S$ are the sets of \emph{accepting states} and \emph{rejecting states}, respectively (hence $S_{acc} \cap S_{rej} = \emptyset$). Let $\Sigma^*$ denote the set of all possible strings over the alphabet $\Sigma$.

Informally, a two-way deterministic finite automaton (2DFA) is  a DFA whose tape head can move to the left, be stationary, or move to the right. A two-way probabilistic finite automaton (2PFA), firstly defined in \cite{def2PFA}, has access to randomness in choosing the next state. 

\begin{definition} \label{def:2DFA-2PFA}
A 2DFA is a tuple $\cM=(S, \Sigma, \delta, s_0, S_{acc}, S_{rej})$.  The input to $\cM$ is of the form $\vertc w\$$ where $w\in \Sigma^*$, and $\vertc, \$ \not\in \Sigma$ are two special symbols  called the left end-marker and right end-marker, respectively.  The transition function $\delta$ is of the following form
\[
\delta: S \times \Big(\Sigma\cup \{\vertc, \$\} \Big) \to S \times \{-1,0,1\}, 
\quad
(s, \sigma) \mapsto (s', d),
\]
which means that when $\cM$ is currently in state $s$ and reads a symbol $\sigma$, its state changes to $s'$ and the tape head moves according to $d$. The value of $d=-1, 0, 1$ corresponds to move to the left, be stationary, and move to the right, respectively. $\cM$ accepts a string $w\in \Sigma^*$  if when $\cM$ runs on $\vertc w \$$, the final state is an accepting state in $S_{acc}$. A language $L \subseteq \Sigma^*$ is recognized by $\cM$ if $\cM$ accepts every $w\in L$ and rejects every $w\not\in L$.

A 2PFA is also a tuple $\cM=(S, \Sigma, \delta, s_0, S_{acc}, S_{rej})$, where the transition function $\delta$ is of the following form
\[
\delta: S \times \Big(\Sigma\cup \{\vertc, \$\} \Big) \times S \times \{-1,0,1\} \to [0,1],
\quad
(s, \sigma, s', d) \mapsto p,
\]
which means that when $\cM$ is currently in state $s$ and reads a symbol $\sigma$, there is a probability $p$ such that its state changes to $s'$ and the tape head moves according to $d$.  $\cM$ accepts  $w \in \Sigma^*$   with probability $1-\epsilon$ if when $M$ runs on $\vertc w \$$, $\Pr[s_f \in S_{acc}] \ge 1-\epsilon$ where $s_f$ is the random final state.  A language $L \subseteq \Sigma^*$ is said to be recognized by $\cM$ with error $\epsilon$  if every $w\in L$ is accepted by $\cM$ with probability at least $1-\epsilon$, and every $w\not\in L$ is accepted by $\cM$ with probability at most $\epsilon$. 
\end{definition}

Quantum finite automata were proposed by   Kondacs and Watrous \cite{2QFA} and also by Moore and Crutchfields \cite{Moo97} and  further studied in \cite{Amb98,bertoni2003quantum,bertoni2006some,bianchi2017quantum,li2015hybrid,nishimura2009application,yamakami2014one}. 
2QCFA were introduced by Ambainis and Watrous \cite{2QCFA} and further studied in \cite{ambainis2009improved,yakaryilmaz2009succinctness,zheng2015power,zheng2013state}.  

Intuitively, a 2QCFA is a 2DFA that has access to  a fixed-size quantum register, upon which it may perform quantum transformations and measurements. Below we give its formal definition, see \cite{2QCFA} for detail.

\begin{definition}[\cite{2QCFA}] \label{def:2QCFA}
A 2QCFA is a tuple $M=(Q, S, \Sigma, \Theta, \delta, q_0, s_0, S_{acc}, S_{rej})$, where $S, \Sigma, \delta, s_0, S_{acc}, S_{rej}$ are similar to 2DFA as in Definition \ref{def:2DFA-2PFA}, whereas $Q$ is a finite set of quantum basis states, $\Theta$ is the quantum transition function that governs the quantum portion of $\cM$, and $q_0$ is the quantum initial state. Specifically, suppose $\cM$ is currently at quantum state $\ket{\phi}$, at classical state $s\in S-S_{acc} - S_{rej}$, and reads a symbol $\sigma$, then $\Theta(s,\sigma)$ can be either a unitary transformation or an orthogonal measurement that acts on $\ket{\phi}$. 
\begin{itemize}
\item[(1)] If $\Theta(s,\sigma)$ is a unitary transformation, then firstly $\cM$ changes  its quantum state  from $\ket{\phi}$ to $\Theta(s,\sigma) \ket{\phi}$, and then changes its classical state from $s$ to  $s'$ and tape head  moves to direction $d$ where $\delta(s,\sigma) = (s',d)$;
\item[(2)] If   $\Theta(s,\sigma)$ is an orthogonal measurement, then $\delta(s,\sigma)$ is a mapping from the set of possible results of the measurement to $S \times \{-1,0,1\}$. $\cM$ firstly collapses its quantum state from $\ket{\phi}$ to a result of the measurement, say $\ket{\varphi}$,  then the classical state and head change according to $\delta(s,\sigma)(\varphi)$.
\end{itemize}

 $\cM$ halts when its classical state reaches an accepting state in $S_{acc}$ or a rejecting state in $S_{rej}$. Since the result of every quantum measurement is probabilistic, the change of the classical state would also be probabilistic. The strings that are accepted by $\cM$ with probability $1-\epsilon$, and the language recognized by $\cM$ with error $\epsilon$, are both defined in the same way as for 2PFA in Definition \ref{def:2DFA-2PFA}.
\end{definition}

Finally, we define the time-space complexity of two-way finite automata.

\begin{definition}  \label{def:TS}
For every two-way finite automaton $\cM$ (can be 2DFA, 2PFA, or 2QCFA), let $\lang(\cM)$ denote the language that it recognizes (possibly with an error $\epsilon\ge 0$). Define $\ttm(\cM)$ to be the time complexity of $\cM$, i.e., the maximal number of steps it takes for $\cM$ to accept a string in $\lang(\cM)$. Define $\ssp(\cM)$ to be the space complexity of $\cM$, i.e., the $\log$ of the number of states of $\cM$ (for 2QCFA, space complexity is the number of qubits plus the $\log$ of the number of classical states). Let $\ts (\cM) = \ttm(\cM) \times \ssp(\cM)$ denote the product of time complexity and space complexity of $\cM$, call it the time-space complexity of $\cM$.

Given a language $L$, define its 2DFA time-space complexity, denoted by $\ts_{2DFA}(L)$,
\begin{equation}   \label{eq:ts-2DFA}
\ts_{2DFA}(L) = \inf_{\cM:\ \lang(\cM) =L} \ts(\cM),
\end{equation}
where $\cM$ ranges over all possible 2DFAs that recognize $L$. 

Define $\ts_{2PFA}(L)$ and $\ts_{2QCFA}(L)$ similarly as \eqref{eq:ts-2DFA}, where the automata ranges over all possible 2PFA and 2QCFA, respectively, with a bounded error. 

Define $\ts_{2QCFA, exact}(L)$ similarly as in \eqref{eq:ts-2DFA}, where the automata ranges over all possible 2QCFA with no error.
\end{definition}

We will be interested in comparing the time-space complexity  of 2DFA, 2PFA, and 2QCFA for some special languages.

Let $w\in \Sigma^*$ and $n= |w|$ be its length. Since the input for two-way finite automata is of the form $\vertc w \$$, it will be convenient to index the symbols of the input via index $0, 1, \ldots, n, n+1$ where $\vertc$ is at position $0$ and $\$$ is at position $n+1$.  This convention will be used in the rest of the paper.

\subsection{Quantum query algorithms} \label{sec:Query}
We assume familiarity with basic notions in quantum computing such as unitary transformation and measurement. Below we briefly define quantum query algorithm and complexity. For detail on classical and quantum query complexity, see \cite{queryICM2018,BdW02, deWolfnotes}. 

Given a function $f: \{0,1\}^n \to \{0,1\}$ and an input $x \in \{0,1\}^n$. A quantum query $\cO_x$ is defined as the following  unitary transformation working on two registers, 
\begin{equation}  \label{eq:quantumquery}
\cO_x: \ket{i, b} \mapsto \ket{i, b\oplus x_i},
\end{equation}
where $i = 1,2,\ldots,n$ and $b\in \{0,1\}$. Sometimes a quantum query algorithm may use another working register, in which case the quantum query $\cO_x$ is defined as $\cO_x \ket{i,b,w} = \ket{i, b\oplus x_i, w}$. A $t$-query algorithm starts from some initial state, say $\ket{\phi}$, and performs unitary transformations and quantum queries successively and reaches a final state $\ket{\varphi} = U_t \cO_x U_{t-1} \cO_x \cdots U_1 \cO_x U_0 \ket{\phi}$, upon which a quantum measurement is performed and outputs accordingly. The \emph{quantum query complexity} of a  function $f$, denoted by $Q_{\epsilon}(f)$, is the minimal number of queries needed for a quantum query algorithm that for every $x$ outputs the value of $f(x)$  with probability of error at most $\epsilon$.  Let $Q(f)$ denote $Q_{1/3}(f)$. Let $Q_E(f)$ denote the \emph{exact} quantum query complexity of $f$, i.e., the algorithm should output the correct value of $f$ with probability $1$. See \cite{Amb13,ambainis2014exact,montanaro2015exact,qiu2020revisiting} for details for exact query complexity.   




\subsection{Communication complexity} \label{def:CC}
We briefly define the standard classical two-party communication  model  introduced by Yao \cite{Yao79}, for related works  see \cite{anshu2016new,jain2016direct,KusNis97,RYCC2020}.

Given a two-party function $f: \{0,1\}^n \times \{0,1\}^n \to \{0,1\}$. Two players Alice and Bob both know the function $f$, Alice receives an input $x$ and Bob receives an input $y$. A communication protocol $\cP$ defines what bits Alice and Bob send to each other so that they both know the value $f(x,y)$ in the end. The communication protocol can be either one round or multiple rounds. Alice and Bob may do some individual computation before sending each bit to the other,  these individual computation are costless, the goal is to minimize the number of bits communicated. For example, the simplest protocol is that Alice sends her input $x$ to Bob, using $n$ bits, then Bob computes $f(x,y)$ and sends this value, using only $1$ bit, to Alice.  The total communication cost of this protocol is $n+1$ bits. 

In a deterministic communication protocol, the individual computation is deterministic, while in a randomized communication protocol,  each party can use their own randomness in their individual computation\footnote{This is called as the private coin communication model.}.  The deterministic communication complexity of a function $f$, denoted by $D^{cc}(f)$, is the least number of bits exchanged in any deterministic protocol that computes $f(x,y)$ correctly for every  input $(x,y)$. The randomized communication complexity of a function $f$ with error $\epsilon$, denoted by $R^{cc}_\epsilon(f)$, is the least number of bits in any randomized protocol  that computes $f(x,y)$ with probability of error at most $\epsilon$ for every input $(x,y)$. Obviously, $R^{cc}_\epsilon(f) \le D^{cc}(f) \le n+1$. We abbreviate $R^{cc}_{1/3}(f)$ simply as $R^{cc}(f)$.


We will  use the following fact on communication complexity of the equality function $\EQ_n$. The equality function $\EQ_n: \{0,1\}^n \times \{0,1\}^n \to \{0,1\}$ is defined as $\EQ_n(x,y)   = 1$ iff $x=y$.


\begin{lemma}[\cite{KusNis97,RYCC2020}]  \label{lem:cc-results}
$D^{cc}(\EQ_n) = n+1$, $R^{cc}(\EQ_n) = \Theta(\log n)$.\end{lemma}


\section{Two-way finite automata and communication protocols}  \label{sec:Simulation}
We will build a connection between two-way finite automata and two-party communication protocols, via the correspondence between the language $L_f(n)$ as defined in \eqref{eq:L-f} and the two-party function $f$. Informally, we show that two-way finite automata and communication protocols can simulate each other in both directions, at least in some special cases.

\subsection{From two-way finite automata to communication protocols}   \label{sec:automaton-to-ccprotocol}

In this section we prove Theorem \ref{thm:lb}. The reason that we gain an extra factor $n$ lies in the fact that we interpolate $\#$ between $x$ and $y$ in defining $L_f(n)$ in \eqref{eq:L-f}.


\begin{lemma}  \label{lem:automaton-to-ccprotocol}
Let $\cM$ be a 2DFA that recognizes $L_f(n)$, then $D^{cc}(f) \le \frac{\ts(\cM)}{n} + 1$. 

Similarly, if $\cM$ is a 2PFA that recognizes $L_f(n)$ with an error $\epsilon$, then $R^{cc}_\epsilon(f) \le \frac{\ts(\cM)}{n} + 1$; if $\cM$ is a 2QCFA that recognizes $L_f(n)$ with an error $\epsilon$, then $Q^{cc}_\epsilon(f) \le \frac{\ts(\cM)}{n} + 1$.
\end{lemma}

\begin{proof}
Recall the definition of $\ssp(\cM)$ and $\ttm(\cM)$ in Definition \ref{def:TS}. 
Consider the following deterministic communication protocol $\pi$. Alice starts running $\cM$ with the partial input $x\#^n$. When the tape head moves out of the partial input from the right (i.e., it wishes to read the first symbol of $y$), Alice sends the current state to Bob, using $\ssp(\cM)$ bits. Note that it is unnecessary for Alice to send the tape head position to Bob since Bob knows the tape head must be at the first position of $y$. Now Bob starts running $\cM$ with the partial input $\#^ny$. Similarly, if the tape head moves out of his input from the left (i.e., it wishes to read the last symbol of $x$), Bob sends the current state of $\cM$ to Alice using $\ssp(\cM)$ bits. They repeat this process until the automaton reaches its final state. If the final state is an accepting state, the current player sends a single bit $1$ to the other player, indicating the simulation of $\cM$ is over and the result is accepting, otherwise the current player sends $0$ indicating the result is rejecting. Since $\cM$ recognizes $L_f(n)$, the protocol $\pi$ computes $f$ correctly. Furthermore, every time when Alice or Bob switches, she or he must have run $\cM$ for at least $n$ steps because  the $\#$ symbol must be read for at least $n$ times. Hence, the communication cost of $\pi$ is at most $\ssp(\cM) \times \frac{\ttm(\cM)}{n} + 1  =\frac{\ts(\cM)}{n} + 1$.  

The proof is the same for 2PFA. For 2QCFA, the proof is also similar: they send both the current classical state and quantum state using qubits, the length of which by definition is at most the space complexity of the 2QCFA. 
\end{proof}

Lemma \ref{lem:automaton-to-ccprotocol} immediately implies Theorem \ref{thm:lb} and the quantum counterparts, summarized below. 

\begin{corollary} \label{cor:automaton-to-ccprotocol}
The following hold,
\begin{itemize}
\item[(1)] $\ts_{2DFA}(L_f(n)) \ge n(D^{cc}(f) - 1)$;
\item[(2)] $\ts_{2PFA}(L_f(n)) \ge n(R^{cc}(f) - 1)$;
\item[(3)] $\ts_{2QCFA}(L_f(n)) \ge n(Q^{cc}(f) - 1)$;
\item[(4)] $\ts_{2QCFA,exact}(L_f(n)) \ge n(Q^{cc}_E(f) - 1)$.
\end{itemize}
\end{corollary}

We now prove Corollary \ref{cor:lb}. For this purpose we briefly discuss lifting theorems in communication complexity. Let 
    $f = h\circ g$
be a composed function of the form \eqref{eq:def-composition}. Suppose the input for $f$ is 
    $(x,y) \in \{0,1\}^n \times \{0,1\}^n$
where 
    $n=pm$.
Write
    $x=(x_1, \ldots, x_p)$ and
    $y=(y_1, \ldots, y_p)$  
where
    $(x_i,y_i) \in \{0,1\}^m \times \{0,1\}^m$.
Recall
    $f(x,y) = h(z)$
where
    $z_i = g(x_i,y_i)$
for $1 \le i \le p$. Given a deterministic query algorithm $\cA$ computing $h$ and a deterministic communication protocol $\tau$ computing $g$, they naturally induce a communication protocol $\pi$ computing $f$, as follows: Alice and Bob together simulate $\cA$,  when $\cA$ queries $z_i$, Alice and Bob use the communication protocol $\tau$ to compute $z_i = g(x_i,y_i)$ (this is possible because Alice knows $x$ hence knows $x_i$ and Bob knows $y$ hence knows $y_i$, for every $1\le i\le p$), and then they jointly go to the next query of $\cA$, and finally, the output of $\cA$ will be the output for $\pi$. It is not hard to see that protocol $\pi$ indeed computes $f = h \circ g$ correctly, and the communication cost of $\pi$ is at most the number of queries in $\cA$ times the communication cost of $\tau$. Hence, one has
    \[
        D^{cc}(f) \le D(h) \cdot D^{cc}(g).
    \]
A  lifting theorem for deterministic communication complexity  says the above protocol $\pi$ is basically the optimal communication protocol for $f = h \circ g$. In other words, it asserts that
    \[
        D^{cc}(f) \ge \Omega(D(h) \cdot D^{cc}(g)).
    \]
Similarly, a lifting theorem for randomized communication complexity asserts that 
    $R^{cc}(f) \ge \widetilde{\Omega}(R(h) \cdot R^{cc}(g))$.

We are now ready to prove Corollary \ref{cor:lb}, restated below.

\corlb*

\begin{proof}
Let 
    $f= h \circ \IP_m$
where $m=\log p$. 
It is well-known (see, e.g., \cite{KusNis97}) that 
    $D^{cc}(\IP_m) = \Theta(m)$ and 
    $R^{cc}(\IP_m) = \Theta(m)$.
\cite{chattopadhyay2019simulation,liftingdet} showed the following  lifting theorem for deterministic communication complexity  
    $D^{cc}(f) = D^{cc}(h \circ \IP_m) = \Omega(D(h) \log p)$.
This together with Corollary \ref{cor:automaton-to-ccprotocol}  implies
    \[
        \ts_{2DFA}(L_f(n)) \ge n(D^{cc}(f) - 1)
        \ge n \cdot \Omega(D(h) \log p) = \widetilde{\Omega}(n D(h)),
    \]
where we used the fact that 
    $n = pm = p\log p$,  
hence
    $\log p = \widetilde{\Theta}(\log n)$.

Chattopadhyay {\it et al.~}\cite{chattopadhyay2019query} showed the following  lifting theorem  for randomized communication complexity 
    $R^{cc}(f) = R^{cc}(h \circ \IP_m) = \Omega(R(h) \log p)$. 
This together with Corollary \ref{cor:automaton-to-ccprotocol} proves \eqref{eq:random-lifting}, similarly as above.
\end{proof}

\subsection{From communication protocols to two-way finite automata}  \label{sec:ccprotocol-to-automaton}
Consider the opposite direction of Lemma \ref{lem:automaton-to-ccprotocol}: does a (classical or quantum) communication protocol for $f$ give rise to a two-way automaton for $L_f(n)$? We show that this can be done in some special cases.

\subsubsection{Simulating query algorithms} \label{sec:simulate-quantum}
In this section we prove  Theorem \ref{thm:ub}. 
For notational clarity of the proof, we consider the simpler case $m=1$ first. Recall that 
given a function $h: \{0,1\}^n \to \{0,1\}$ and a two-party function $g: \{0,1\} \times \{0,1\} \to \{0,1\}$, they define $f=h \circ g$ as follows,
    \begin{equation}  \label{eq:composition}
        f(x,y) = h(g(x_1,y_1), g(x_2,y_2), \ldots, g(x_n,y_n)).
    \end{equation}
Buhrman et al. \cite{Buh98} showed that $Q^{cc}(f) \le O(Q(h) \log n) = \widetilde{O}(Q(h))$. We focus on this special type of communication protocols for $f$, i.e.,  protocols that are obtained via query algorithms for $h$. Below we show that given a quantum query algorithm for $h$, one can construct a 2QCFA recognizing $L_f(n)$ with the same error. The simulation that we construct is inspired by \cite{Buh98}. 

We restate Theorem \ref{thm:ub} in the case $m=1$ in a more precise form in the following.

\begin{thm}   \label{thm:m-is-1}
    Let $f = h \circ g$ be of the form \eqref{eq:composition}. Suppose that there is a quantum query algorithm $\cA$ solving $h$ with error $\epsilon$, furthermore, suppose that $\cA$ uses $k$ quantum basis states and  $t$ queries. Then, there exists a 2QCFA $\cM$ that recognizes $L_f(n)$ with error $\epsilon$. Furthermore, $\cM$ uses $2k$ quantum basis states, $O(t\cdot n)$ classical states, and $\ttm(\cM)=O(t\cdot n)$. As $t = O(n)$,  
    \[
    \ts(\cM) = O\Big(tn \cdot (\log(tn) + \log 2k)\Big) = O\Big(tn \cdot (\log n + \log k)\Big) .
    \]
\end{thm}

\begin{proof}
Given $(x,y) \in \{0,1\}^n \times \{0,1\}^n$, let $z=(z_1, \ldots, z_n) \in \{0,1\}^n$ where 
\begin{equation}   \label{eq:z_i}
z_i = g(x_i, y_i).
\end{equation}
Then $f(x,y) = h(z)$. 

Consider the quantum query algorithm $\cA$ for $h$, the input to $\cA$ is $z$. Suppose $\cA$ is of the following form with an initial quantum state $\ket{\phi_0}$,
\begin{equation} \label{eq:query-alg-A}
\begin{tikzcd}   
\ket{\phi_0} \arrow[r, "U_0"]
& 
\ket{\phi_1} \arrow[r, "\cO_z"]
&
\ket{\phi'_1} \arrow[r, "U_1"]
&
\ket{\phi_2} \arrow[r, "\cO_z"]
&
\ket{\phi'_2} \arrow[r, "U_2"]
& 
 \cdots
& 
\ket{\phi_t} \arrow[r, "\cO_z"]
&
\ket{\phi'_t} \arrow[r, "U_t"]
&
\ket{\phi_{t+1}}.
\end{tikzcd}
\end{equation}
Consider the first segment of $\cA$, 
\begin{equation} \label{eq:query-alg-A-1}
\begin{tikzcd}   
\ket{\phi_0} \arrow[r, "U_0"]
& 
\ket{\phi_1} \arrow[r, "\cO_z"]
&
\ket{\phi'_1}.
\end{tikzcd}
\end{equation}
Observe that the structure of $\cA$ is simply a repetition of the above segment $t$ times, each time with possibly a different unitary transformation. Finally, $\cA$ makes one more unitary transformation $U_t$ and then makes a measurement on $\phi_{t+1}$ and outputs accordingly.  

We now construct a 2QCFA $\cM$ that recognizes $L_f(n)$ (with the same error as $\cA$ for $h$). Recall that $w\in L_f(n)$ is of the form $w=x\#^ny$ where $(x,y) \in \{0,1\}^n \times \{0,1\}^n$, the input to $\cM$ is $\vertc w\$$. Assume for simplicity that the tape is circular, i.e., after reading $\$$, if the head moves to the right then it reaches $\vertc$. 

Firstly, it is easy (using classical states only) to check that $w$ is of the form $x\#^n y$, otherwise reject the input. This step uses at most $O(n)$ classical states and $O(n)$ time. Now assume the input is in this form. 

We focus on simulating the first segment \eqref{eq:query-alg-A-1}, the rest can be done similarly. Note that the unitary operation $U_0$ does not depend on the input $z$, hence it can be applied in the 2QCFA directly, as follows.

\begin{equation} \label{eq:2QCFA-M-U0}
\begin{tikzcd}   
\ket{\phi_0,0} \arrow{rr}{\Theta(s_{1,0},\vertc)}
& &
\ket{\phi_1,0}
\\
s_{1,0} \arrow{rr}{\vertc} 
&&
s_{1,1}
\end{tikzcd}
\end{equation} 
where the transition functions are defined by
\begin{equation}   \label{eq:def-step-U0}
\begin{cases}
\Theta(s_{1,0},\vertc) = U_0 \otimes I, \\
\delta(s_{1,0}, \vertc) = (s_{1,1}, 1).
\end{cases}
\end{equation}
That is, the quantum step applies $U_0$ to $\ket{\phi_0}$ and keeps the last qubit unchanged. The classical step changes the state to $s_{1,1}$ and  the tape head moves to the right.

Next, $\cM$ will read the input twice to simulate the query $\cO_z: \ket{i, b_i} \mapsto \ket{i, z_i \oplus b_i}$. It is for this purpose that we use an extra working qubit to help us. So, the query that we will simulate is $\cO_z \otimes I: \ket{i, b_i,0} \mapsto \ket{i, z_i \oplus b_i, 0}$. Since $\cO_z \ket{\phi_1} =\ket{\phi'_1}$, we have $(\cO_z \otimes I) \ket{\phi_1, 0} = \ket{\phi'_1, 0}$ as desired. We use five steps to describe this simulation. In all these steps, the tape head keeps moving to the right.
\begin{itemize}
\item[(1)] $\cM$ reads $x$ the first time. The transition of quantum and classical states are as follows,
\begin{equation} \label{eq:2QCFA-M-step1}
\begin{tikzcd}   
\ket{\phi_1,0} \arrow{rr}{\Theta(s_{1,1},x_1)}
&&
\ket{\varphi_2} \arrow{rr}{\Theta(s_{1,2},x_2)}
&&
\cdots
\ket{\varphi_{n}} \arrow{rr}{\Theta(s_{1,n},x_n)}
&&
\ket{\varphi_{n+1}} 
\\
s_{1,1} \arrow{rr}{x_1}
&&
s_{1,2} \arrow{rr}{x_2}
&&
\cdots
s_{1,n} \arrow{rr}{x_n}
&&
s_{1,n+1}
\end{tikzcd}
\end{equation}
where for $j= 1, \ldots, n$, 
\begin{equation}   \label{eq:def-step-1}
\Theta(s_{1,j},x_j) \ket{i,b_i,c}
=
\begin{cases}
\ket{i,b_i,x_i\oplus c}, &\quad i=j, \\
\ket{i,b_i,c}, &\quad i\neq j.
\end{cases}
\end{equation}
Note that after this step, the state $\ket{i,b_i,0}$ in $\ket{\phi_1, 0}$ has changed to $\ket{i,b_i,x_i}$ in $\ket{\varphi_{n+1}}$ for every $i=1,2,\ldots, n$. In other words, after reading $x$ the automaton $\cM$ simulates $\cO_x$ acting on the last qubit. 

\item[(2)] When $\cM$ reads $\#$, it keeps its quantum state unchanged, and classical state changes from $s_{1,n+1}$ to $s_{1,n+2}$, $s_{1,n+3}$ etc until $s_{1,2n+1}$ after $\cM$ reads the last $\#$.  For notational uniformity, we denote the quantum state after reading the last $\#$ by $\ket{\varphi_{2n+1}} = \ket{\varphi_{n+1}}$.

\item[(3)] $\cM$ reads $y$ the first time. Similar to reading $x$, this time $\cM$ simulates $\cO_y$ but acting on the penultimate qubit, and also translated by $x$ using the function $g$. Specifically, 
\begin{equation} \label{eq:2QCFA-M-step1}
\begin{tikzcd}   
\ket{\varphi_{2n+1}} \arrow{rr}{\Theta(s_{1,2n+1},y_1)}
&&
\ket{\varphi_{2n+2}} \arrow{rr}{\Theta(s_{1,2n+2},y_2)}
&&
\cdots
\ket{\varphi_{3n}} \arrow{rr}{\Theta(s_{1,3n},y_n)}
&&
\ket{\varphi_{3n+1}} 
\\
s_{1,2n+1} \arrow{rr}{y_1}
&&
s_{1,2n+2} \arrow{rr}{y_2}
&&
\cdots
s_{1,3n} \arrow{rr}{y_n}
&&
s_{1,3n+1}
\end{tikzcd}
\end{equation}
where for $j = 1, \ldots, n$,
\begin{equation}   \label{eq:def-step-1}
\Theta(s_{1,2n+j},y_j) \ket{i,b_i,c}
=
\begin{cases}
\ket{i,g(c,y_i)\oplus b_i,c}, &\quad i=j, \\
\ket{i,b_i,c}, &\quad i\neq j.
\end{cases}
\end{equation}
Since in $\ket{\varphi_{2n+1}} = \ket{\varphi_{n+1}}$, the basis state we are interested in is in state $\ket{i,b_i,x_i}$ for every $i=1, \ldots, n$. Hence, after Step 3, in $\ket{\varphi_{3n+1}}$ this state has changed to $\ket{i,g(x_i, y_i)\oplus b_i,x_i} = \ket{i,z_i\oplus b_i,x_i}$. 

\item[(4)] The tape moves to the right and goes back to $\#$, then it keeps moving to the right and reads $x$ the second time. This step is similar to Step (1): the automaton $\cM$ simulates $\cO_x$ one more time, acting on the last qubit. As a result, the state we are interested in  changes from $\ket{i,z_i\oplus b_i,x_i}$ to $\ket{i,z_i\oplus b_i,x_i\oplus x_i} = \ket{i,z_i\oplus b_i,0}$, as desired. At the end of this step, we have simulated $\cO_z \otimes I$. 

\item[(5)] The $\cM$ keeps moving to the right, with quantum state unchanged, until it returns to $\vertc$, where it changes its classical state to $s_{2,0}$. Note that the final quantum state is in $\ket{\phi'_1,0}$ as desired. 
\end{itemize}

We have shown that the first segment (i.e., \eqref{eq:query-alg-A-1}) of query algorithm $\cA$  can be simulated with number of classical states $O(n)$ and time $O(n)$. Obviously, the automaton $\cM$ could repeat this simulation $t$ times, mimicing the query algorithm $\cA$. In the end $\cM$ makes the same measurement as $\cA$, and accepts iff $\cA$ outputs $1$ according to the measurement result. Since $\cM$ simulates the query algorithm $\cA$ with input $z=(g(x_1, y_1), \ldots, (x_n,y_n))$, it is obvious that $\cM$ recognizes $L_f(n)$ with the same error as $\cA$ does. In total, the classical states used by $\cM$ is $O(t\cdot n)$ and $\ttm(\cM) = O(t\cdot n$).  Since $\cM$ uses one extra qubit, the number of its quantum basis states is $2k$.  
\end{proof}

It is not hard to see that a similar proof for Theorem \ref{thm:m-is-1} would prove Theorem \ref{thm:ub}. We omit the full detail but point out one difference,  that when $g$ has $m$ input bits, the 2QCFA $\cM$ would use $2^m k$ quantum basis states, hence the additive factor $\log k$ in Theorem \ref{thm:m-is-1}  becomes $\log (2^m k) = m+\log k = m + O(\log n)$.

\subsubsection{Simulating communication protocols directly} \label{sec:simulate-cc-EQ}

To prove Theorem \ref{thm:sep-L-EQ}, we establish the following lemma first.  The proof can be viewed as an example that one can simulate some communication protocols \emph{directly} by two-way finite automata.

\begin{lemma}  \label{lem:ccprotocol-to-automaton-EQ}
$\ts_{2PFA}(L_{\EQ}(n)) = O(n \log n)$. 
\end{lemma}

\begin{proof}
    We will turn a randomized communication protocol to a 2PFA. Let us briefly recall the standard randomized communication protocol $\pi$ for solving $\EQ_n$ with a bounded error, for detail see \cite{KusNis97}. Alice samples a random prime number $p \le n^2$. Viewing $x \in \{0,1\}^n$ as an integer $x \in \{0,1,\ldots, 2^n-1\}$, she  sends both $p$ and $x \mod p$ to Bob, using $O(\log n)$ bits. Bob sends $1$ back to Alice if $y \equiv x \mod p$, otherwise he sends $0$. The bit that Bob sends is the output of the protocol.
    
    We now turn this protocol to  a 2PFA $\cM$ recognizing $L_{\EQ}(n)$ with a bounded error, as follows. Let $w$ be an input string and $\vertc w\$$ is the input to $\cM$. $\cM$ firstly checks whether $w$ is of the correct form $x\#^ny$, using time $O(n)$ and number of states $O(n)$.  Then $\cM$ returns to the $\vertc$ symbol and start simulating $\pi$. $\cM$ firstly reads $x$ deterministically, then on reaching the first $\#$ symbol it changes its state randomly to a state $s_{p,a,0}$ where $p\le n^2$ is a random prime number and $x \equiv a \mod p$. It moves its head to the right and keeps reading $\#$ and $y$ while keeping at the state $s_{p,a,0}$. Until when it reaches $\$$, it changes its state to $s_{p,a,b}$ where $y\equiv b \mod p$. The states $s_{p,a,b}$ are accepting states if $a=b$ and rejecting states otherwise. Obviously, the time it takes to simulate $\pi$ is $O(n)$ and the number of states is $O(n^6)$. Hence,
    \[
    \ttm(\cM) = O(n) + O(n) = O(n), \quad
    \ssp(\cM) = O(n) + O(n^6) = O(n^6). 
    \]
    Therefore, $\ts(\cM) = O(n \log n)$. 
    
    The fact that $\cM$ recognizes $L_{\EQ}(n)$ with a bounded error follows from the same analysis for showing the communication protocol $\pi$ has a bounded error. For completeness, we record the analysis here. It is obvious that $\cM$ always accepts $w\in L_{\EQ}(n)$. From now on, we assume $w\not\in L_{\EQ}(n)$ and $\cM$ accepts $w$, or equivalently, $x\neq y$ but there exists some prime numbers $p\le n^2$ such that $x\equiv y \mod p$. Consider the set 
    \[
    Bad(x,y) = \{2\le p \le n^2, p \text{\ is prime}\ |\ x\equiv y \mod p\}. 
    \]
    Then, 
    \begin{equation}  \label{eq:error}
    \Pr[\cM \text{\ accepts\ } x\#^ny] = \frac{|Bad(x,y)|}{\text{number of primes}\le n^2}.
    \end{equation}
    Let $k = |Bad(x,y)|$. Suppose that the prime numbers in the set $Bad(x,y)$ are $p_1, p_2, \ldots, p_k$, then $x \equiv y \mod p_i$ for every $1 \le i\le k$. Hence, 
    \[
    x \equiv y \mod p_1 p_2 \cdots p_k.
    \]
    However, $x\neq y$, hence, $0< |x - y| \le 2^n - 1$. As a result, $2^k\le p_1 p_2 \cdots p_k \le |x-y| \le 2^n - 1$, implying $k \le n-1$. On the other hand, the prime number theorem says that the number of primes less than $n^2$ approximates $\frac{n^2}{\log n^2} = \frac{n^2}{2 \log n} \gg n-1 \ge k$ when $n$ is sufficiently large. This implies $\lim_{n\to \infty} \Pr[\cM \text{\ accepts\ } w] = 0$. In fact, since the number of primes less than $10^2$ is $25$, one has for $n=10$, $\Pr[\cM \text{\ accepts\ } x\#^ny] \le \frac{10-1}{25} =0.36$. For every $n< 10$, obviously one can construct a constant time-space 2PFA (in fact even 2DFA) to recognize $L_{\EQ}(n)$ with no error. Hence, the lemma holds for every $n$. 
\end{proof}

Now we are ready to prove Theorem \ref{thm:sep-L-EQ}, restated below. 

\thmLEQ*

\begin{proof}
    For 2DFA, the lower bound follows from Lemma \ref{lem:cc-results} and Corollary \ref{cor:automaton-to-ccprotocol}, $\ts_{2DFA}(L_{\EQ}(n)) \ge n(D^{cc}(\EQ_n) - 1) = n^2$. For the upper bound, consider the 2DFA $\cM$ that firstly memorizes $x$, then it checks the number of $\#$ is $n$, and then compares $y$ to $x$. Hence, $\lang(\cM) = L_{\EQ}(n)$. It is easy to see that  $\cM$ can be constructed in the form of a tree structure, resulting $\ttm(\cM) = 3n$, and number of states $O(2^{3n})$. Hence $\ts(\cM) = 3n \times \log O(2^{3n}) = O(n^2)$.  
    
    For 2PFA, similarly the lower bound is  $\ts_{2PFA}(L_{\EQ}(n)) \ge n(R^{cc}(\EQ_n) - 1) \ge  \Omega(n \log n)$. The upper bound comes from Lemma \ref{lem:ccprotocol-to-automaton-EQ}.
\end{proof}

\begin{remark}
(1)  Although we demonstrated with only one example (i.e., the equality function), it might be possible that some other classical or quantum communication protocols could also be simulated directly by two-way finite automata of corresponding types. 
(2) Independently from simulating communication protocols as we show here in Lemma \ref{lem:ccprotocol-to-automaton-EQ}, the idea of reducing a general computation (e.g., whether $x=y$) to a modulo computation (e.g., whether $x \equiv y \mod p$) was also \emph{directly} used in studying 2PFA. The proof for that the non-regular language $\{a^n b^n \ | \ n\ge 1\}$ can be recognized by 2PFA is such an example, as shown in \cite{2PFAub}. 
\end{remark}

\section{Separation results for time-space complexity and examples}   \label{sec:sep-results-proof}

We are now ready to prove Theorem \ref{thm:separation-via-query}, restated below. 

\thmSepViaQuery*

\begin{proof}
    Consider (1) for example. Ambainis {\it et al.~}\cite{ambainis2017separations} showed that there exists a function $h: \{0,1\}^p \to \{0,1\}$ such that,  
        $D(h) \ge \widetilde{\Omega}(p)$ and 
        $Q(h) \le \widetilde{O}(p^{1/4})$.
    Let 
        $f=h\circ \IP_m$
    where
        $m=\Theta(\log p)$.
    Let
        $n=pm$
    be the input size for $f$. 
    
    By Corollary \ref{cor:lb}, 
        \[
            \ts_{2DFA}(L_f(n)) \ge  \widetilde{\Omega}(n D(h)) \ge  \widetilde{\Omega}(n p) = \widetilde{\Omega}(n^2),
        \]
    where we used $p = \widetilde{\Theta}(n)$ because $m=\Theta(\log p)$. 
    Furthermore, for every language $L$ that consists of length-$k$ strings on an alphabet of constant size, it is easy to see that
        $\ts_{2DFA}(L) \le O(k^2)$.
    Hence,  
        $\ts_{2DFA}(L_f(n)) \le O(n^2)$.
    Therefore,
        $\ts_{2DFA}(L_f(n)) = \widetilde{\Theta}(n^2)$. 
    
    By Theorem \ref{thm:ub},
        \[
            \ts_{2QCFA}(L_f(n)) \le  \widetilde{O}(n Q(h)) \le  \widetilde{O}(n p^{1/4}) = \widetilde{O}(n^{5/4}).
        \]
        
    The proofs for (2), (3) and (4) can be similarly established by applying the following separation results for $h$.     
        
    Ambainis {\it et al.~}\cite{ambainis2017separations} showed that, for (2) there exists $h$ such that $D(h) \ge \widetilde{\Omega}(p)$ and $Q_E(h) \le \widetilde{O}(p^{1/2})$; and for (4) there exists $h$ such that $R(h) \ge \widetilde{\Omega}(p)$ and $Q_E(h) \le O(p^{2/3})$.

    For (3),  Sherstov {\it et al.~}\cite{sherstov2020optimal} introduced a function $h: \{0,1\}^p \to \{0,1\}$ such that  $R(h) \ge \widetilde{\Omega}(p)$ and $Q(h) \le \widetilde{O}(p^{1/3})$. 
\end{proof}

Unfortunately, the definitions of the functions $h$ used in Theorem \ref{thm:separation-via-query} are all too complicated to be defined here. We refer the interested reader to the original papers \cite{ambainis2017separations,sherstov2020optimal}. Below we give two examples where the functions involved are simple, though the separations obtained from them are weaker than Theorem \ref{thm:separation-via-query}.

To define a composed function of the form \eqref{eq:def-composition}, we need to specify both $h$ and $g$. In both examples the two-party function $g$ is simply the $\wedge$ function, 
    $g=\wedge: \{0,1\} \times \{0,1\} \to \{0,1\}$, 
i.e., $g(a,b) = 1$ iff $a=b=1$. 
We proceed to define the two $h$s. 
The first function is the 
    $\OR_n: \{0,1\}^n \to \{0,1\}$, 
i.e., $\OR_n(x) = 0$ iff $x_i=0$ for every $1\le i \le n$.   
To define the second function 
    $\NE_n$ for $n=3^d$ 
where $d\in \N$, 
we start by defining 
    $\NE^1: \{0,1\}^3 \to \{0,1\}$ 
as follows:  
    $\NE^1(x_1, x_2, x_3) = 0$ if $x_1 = x_2 = x_3$, 
otherwise 
    $\NE^1(x_1, x_2, x_3) = 1$. 
The function    
    $\NE_n: \{0,1\}^n \to \{0,1\}$ 
is defined by composing $\NE^1$ with itself $d$ times, i.e.,
    \[
        \NE_n(x) = \NE^d(x) 
        = \NE^1(\NE^{d-1}(x_1, \ldots, x_{3^{d-1}}), \NE^{d-1}(x_{3^{d-1}+1}, \ldots, x_{2 \cdot 3^{d-1}}), \NE^{d-1}(x_{2\cdot 3^{d-1}+1}, \ldots, x_{3^d})).
    \]
Hence, the two composed functions are the following. The first one is the so-called intersection function 
    $\INTS_n = \OR_n \circ \wedge$.
The second one is
    $\RNE_n = \NE_n \circ \wedge$.
The corresponding two languages are the following.
    \begin{equation}  \label{eq:L-INTS}
        L_{\INTS}(n) = \{x\#^ny \ | \ x, y \in \{0,1\}^n, \text{\ there exists some coordinate}\ i \text{\ such that}\ x_i=y_i=1\}, 
    \end{equation}
and
    \begin{equation}  \label{eq:L-RNE}
        L_{\RNE}(n) = \{x\#^ny \ | \ x, y \in \{0,1\}^n, \RNE_n(x,y) = 1\}.
    \end{equation}

\begin{thm}  \label{thm:sep-L-INTS-L-RNE}
    The following hold.
    \begin{enumerate}[(i)]
        \item $\ts_{2PFA}(L_{\INTS}(n)) = \Omega(n^2)$,              
        $\ts_{2QCFA}(L_{\INTS}(n)) = O(n^{3/2} \log n)$.
        
        \item $\ts_{2PFA}(L_{\RNE}(n)) = \Omega(n^2)$,              
        $\ts_{2QCFA, exact}(L_{\RNE}(n)) = O(n^{1.87} \log n)$.
    \end{enumerate}
\end{thm}

\begin{proof}
(i) The lower bound follows from Theorem \ref{thm:lb} and 
    $R^{cc}(\INTS_n) = \Omega(n)$ (see e.g., \cite{KusNis97}).
The upper bound follows from Theorem \ref{thm:ub} and the fact that Grover \cite{Gro96} showed 
    $Q(\OR_n) = O(\sqrt{n})$.
(ii) Similarly, this follows from Ambainis \cite{Amb13} where it was shown that 
    $R^{cc}(\RNE_n) = \Omega(n)$
and
    $Q_{E}(\NE_n) = O(n^{0.87})$.
\end{proof}

\section{Discussion}   \label{sec:conclusion}

We developed a technique that, when used together with lifting theorems in communication complexity, automatically  lifts separations of query complexity to separations of time-space complexity of two-way finite automata. We use this lifting technique to obtain a collection of new time-space separations for two-way finite automata in Theorem \ref{thm:separation-via-query}. Our results can be viewed as an application of lifting theorems in communication complexity to the complexity of automata. Since the lower bound in Corollary \ref{cor:automaton-to-ccprotocol} holds for $Q_E^{cc}$ as well, our technique in theory could be used to derive a time-space separation between $\ts_{2QCFA,exact}$ and $\ts_{2QCFA}$. This would require a currently unknown lifting theorem for quantum communication complexity. 

A limitation of our lifting technique is that the separation result obtained for time-space complexity for recognizing $L_f(n)$ is not necessarily tight, even if the corresponding separation in query complexity for $h$ is tight. It is possible to improve some of the separation results in Theorem \ref{thm:separation-via-query}  by directly simulating certain communication protocols  as we demonstrated in Theorem \ref{thm:sep-L-EQ}. The lifting technique we developed also demonstrates that why a lifting theorem with constant-size gadgets (i.e., do not require $m =\Theta(\log p)$ in Corollary \ref{cor:lb} but allow any $m \ge \Omega(1)$) would be more desirable, as we discussed in Section \ref{sec:lift-limitation}.

In studying time-space tradeoffs for Turing machines and branching programs, the multiparty number on the forehead communication model (see \cite{KusNis97}) has been successfully used, see e.g., \cite{BNS92, cctoBP3}. It might be interesting to see if our technique could be extended to those settings.

\section*{Acknowledgement}

We thank the anonymous referees whose valuable comments led to a much improved version of the paper.
S. Z.  thanks  A.~Ambainis, C.~Mereghetti, B.~Palano and A. Anshu for  discussions. This work was supported by  the  
Major Key Project of PCL, the  National Natural Science Foundation of China (Nos. 62361021,62272492), Guangxi Science and Technology Program (No.GuikeAD21075020),   the Guangdong Basic and Applied Basic Research Foundation (No. 2020B1515020050),  the Innovation Program for Quantum Science and Technology (2021ZD0302900).

\bibliography{mybib}

\end{document}